\documentclass[aps,prd,twocolumn,superscriptaddress,nofootinbib,longbibliography]{revtex4-2}

\usepackage[T1]{fontenc}
\usepackage[utf8]{inputenc}
\usepackage{amsmath,amssymb,amsthm,bm}
\usepackage{mathrsfs}
\usepackage{hyperref}
\usepackage{microtype}
\usepackage{xcolor}

\hypersetup{colorlinks=true,linkcolor=blue,citecolor=blue,urlcolor=blue}

\newcommand{\dd}{\mathrm{d}}
\newcommand{\Lie}{\mathcal{L}}
\newcommand{\cL}{\mathbf{L}}
\newcommand{\btheta}{\boldsymbol{\theta}}
\newcommand{\bomega}{\boldsymbol{\omega}}
\newcommand{\bJ}{\mathbf{J}}
\newcommand{\bQ}{\mathbf{Q}}
\newcommand{\bE}{\mathbf{E}}
\newcommand{\bC}{\mathbf{C}}

\newcommand{\Etot}{\mathbb{E}}
\newcommand{\Ebulk}{\mathcal{E}^{\rm bulk}}
\newcommand{\Vth}{V_{\rm th}}
\newcommand{\SW}{S_{\rm W}}
\newcommand{\Xref}{X_0}
\newcommand{\PhiH}{\Phi_{\alpha}}
\newcommand{\SV}{\mathcal{S}_V}
\newcommand{\ph}{\phi}
\newcommand{\scrP}{\mathscr{P}}
\newcommand{\scrC}{\mathscr{C}}
\newcommand{\scrG}{\mathscr{G}}

\newcommand{\calB}{\mathcal{B}}

\newcommand{\calR}{\mathcal{R}}

\newcommand{\scrZ}{\mathscr{Z}}
\newcommand{\scrT}{\mathscr{T}}
\newcommand{\scrB}{\mathscr{B}}

\newtheorem{proposition}{Proposition}

\begin{document}

\title{Reverse Isoperimetry from Boundary--Completed Canonical Energy}

\author{Naman Kumar}
\email{namankumar5954@gmail.com, naman.kumar@iitgn.ac.in}
\affiliation{Department of Physics, Indian Institute of Technology Gandhinagar, Palaj, Gujarat 382355, India}

\date{\today}

\begin{abstract}
The reverse isoperimetric conjecture states that, at fixed thermodynamic volume, Schwarzschild--AdS black holes maximize entropy. We show that its local fixed-volume form admits a boundary-completed canonical-energy interpretation. The key observation is that the bulk Hollands--Wald canonical energy vanishes on exact stationary black-hole families, while the fixed-volume entropy Hessian is encoded in an asymptotic boundary contribution. The same structure governs the global comparison. We choose an affine parameterization of a fixed-volume stationary branch connecting Schwarzschild--AdS to the solution of interest. If the stationary boundary Hessian is non-negative along this comparison, entropy concavity gives the reverse isoperimetric bound, while equality is determined by the directions along which the boundary Hessian vanishes. In Einstein gravity this reproduces the area--volume inequality on the stated comparison branches, and the formulation extends naturally to Wald entropy in higher-derivative theories.
\end{abstract}

\maketitle

\emph{\textbf{Introduction}.---}
Extended black-hole thermodynamics promotes the cosmological constant to a
pressure,
\begin{equation}
P=-\frac{\Lambda}{8\pi G},
\end{equation}
and identifies the conjugate quantity as a thermodynamic volume
\cite{Kastor:2009wy,Dolan:2010ha}.  In Einstein gravity, the reverse
isoperimetric inequality (RII) of Ref.~\cite{Cvetic:2010jb} is
\begin{equation}
\calR\equiv
\left(\frac{(D-1)\Vth}{\omega_{D-2}}\right)^{1/(D-1)}
\left(\frac{\omega_{D-2}}{A}\right)^{1/(D-2)}\geq1,
\label{eq:rii}
\end{equation}
where $A$ is the horizon area and $\omega_{D-2}$ is the area of the unit
sphere.  Equivalently, at fixed thermodynamic volume, the \emph{entropy} is bounded
above by that of the Schwarzschild--AdS black hole with the same
$\Vth$. This is the sense in which Schwarzschild--AdS is the entropy
maximizer.  Static charged black holes, including the gauged-supergravity
examples studied in Ref.~\cite{Cvetic:2010jb}, are admissible competitors:
generically they obey $\calR>1$ and hence have lower entropy at the same
$\Vth$.  Special RN--AdS or equal-charge limits may also saturate
$\calR=1$ because their thermodynamic volume reduces to the geometric ball
volume; these equality-degenerate cases belong to the equality sector, not
to a violation of the RII.

This entropy formulation is the natural one beyond Einstein gravity.  In a
general diffeomorphism-invariant theory, the entropy is not necessarily
$A/4G$; for stationary horizons it is the Noether-charge entropy of Wald and
Iyer~\cite{Wald:1993nt,Iyer:1994ys}. The covariant form of the RII is therefore
\begin{equation}
\SW(X)\leq \SW\!\left(\Xref(\Vth)\right),
\label{eq:generalRII}
\end{equation}
where $X$ is an admissible stationary AdS black hole, $\Xref(\Vth)$ denotes
the Schwarzschild--AdS reference saddle with the same thermodynamic volume,
and $\SW$ is the stationary Iyer--Wald entropy. Possible
equality-degenerate static branches are not included in $\Xref$ itself. Rather,
they are included in the equality sector defined below.
In Einstein
gravity, Eq.~\eqref{eq:generalRII} reduces to Eq.~\eqref{eq:rii}.  A proof within Einstein--AdS gravity using a complementary geometric--analytic approach was given in Ref.~\cite{Kumar:2025mvj}.

The central point of this Letter is that reverse isoperimetry is not
controlled by the bulk Hollands--Wald canonical energy alone.  Along exact
stationary black-hole families the perturbation is time independent,
\begin{equation}
\Lie_\chi v=0,
\end{equation}
so the bulk symplectic integral vanishes.  The second variation of the entropy along the fixed-volume stationary branch relevant to
the RII is instead carried by a constrained asymptotic charge Hessian.  The
proper covariant object is therefore a \emph{boundary-completed canonical energy}.
Its bulk sign is supplied by canonical-energy stability, while the stationary boundary sign is determined by the constrained fixed-volume entropy Hessian and must be checked on the stationary branch under consideration. This distinction becomes important in the global RII comparison, since it is performed entirely within the stationary solution space, where the bulk canonical energy vanishes along the comparison segment. The boundary completion is therefore not a secondary correction to the bulk stability functional. In fact, it is the only nontrivial part of the completed energy that can encode the fixed-volume entropy change between stationary black holes.

\emph{\textbf{Covariant variables}.---}
Let the theory be defined by a diffeomorphism-invariant $D$-form Lagrangian
\begin{equation}
\cL=\cL(\ph;P,\lambda_I),
\end{equation}
where $\ph$ denotes all dynamical fields and $\lambda_I$ denotes possible higher-derivative couplings.  Its first variation is
\begin{equation}
\delta \cL=\bE(\ph)\delta\ph+
\dd\btheta(\ph,\delta\ph).
\label{eq:varL}
\end{equation}
The symplectic current is
\begin{equation}
\bomega=\delta_1\btheta(\ph,\delta_2\ph)
-\delta_2\btheta(\ph,\delta_1\ph).
\label{eq:symp}
\end{equation}
For a vector field $\xi$, the Noether current satisfies
\begin{equation}
\bJ_\xi=\btheta(\ph,\Lie_\xi\ph)-\xi\!\cdot\!\cL
=\dd\bQ_\xi+\xi^a\bC_a,
\label{eq:noethercurrent}
\end{equation}
with $\bC_a=0$ on shell.  For a stationary black hole with bifurcate Killing horizon generated by
\begin{equation}
\chi=t+\Omega_i\varphi_i,
\end{equation}
the entropy is
\begin{equation}
\SW=\frac{2\pi}{\kappa}\int_{\calB}\bQ_\chi,
\label{eq:entropy}
\end{equation}
equivalently written in the usual local Wald-density form on the bifurcation surface $\calB$.  The extended Iyer--Wald first law gives~\cite{Xiao:2023lap}
\begin{equation}
\begin{split}
\delta M={}&T\delta\SW+\Omega_i\delta J_i
+\PhiH\delta Q_\alpha \\
&+\Vth\delta P+\Psi_I\delta\lambda_I .
\end{split}
\label{eq:firstlaw}
\end{equation}
Thus the volume in Eq.~\eqref{eq:generalRII} is not a naive geometric volume, but
\begin{equation}
\Vth=\left(\frac{\partial M}{\partial P}\right)_{\SW,J_i,Q_\alpha,\lambda_I} .
\label{eq:volume}
\end{equation}
In the Euclidean grand-canonical ensemble this is equivalently
\begin{equation}
I_E=\beta G,\qquad
\Vth=\left(\frac{\partial G}{\partial P}\right)_{T,\Omega_i,\PhiH,\lambda_I} .
\label{eq:euclideanvolume}
\end{equation}

\emph{\textbf{Admissible phase space and theorem}.---}
Let $\widetilde{\scrP}$ be the stationary asymptotically AdS solution space,
and let $\scrG$ denote diffeomorphisms and gauge transformations that vanish
sufficiently fast at the AdS boundary or act trivially on all covariant
charges.  The physical phase space is $\scrP=\widetilde{\scrP}/\scrG$, and
the fixed-volume component is
\begin{equation}
\scrC_V=\{X\in\scrP: \Vth(X)=V\}.
\label{eq:CVmain}
\end{equation}
We assume that $\scrC_V$ contains the Schwarzschild--AdS reference saddle
$\Xref(V)$ with the same thermodynamic volume.  This is the round static
saddle that appears in the original reverse isoperimetric conjecture: at
fixed $\Vth$, the entropy is maximized by Schwarzschild--AdS
\cite{Cvetic:2010jb}.  Static charged black holes in gauged supergravity
are admissible competitors in $\scrC_V$.  The checks of
Ref.~\cite{Cvetic:2010jb} show that these charged examples obey the RII. They
generically have $\mathcal R>1$ and therefore lower entropy than
Schwarzschild--AdS at the same thermodynamic volume.  Special RN--AdS or
equal-charge limits can saturate $\mathcal R=1$ because their thermodynamic
volume reduces to the geometric volume equal to the volume of a Euclidean ball in $\mathbb{R}^{d-1}$.

We denote by $\scrZ_V\subset\scrC_V$ the allowed equality sector.  It
contains the equivalence class represented by $\Xref(V)$ and, when present,
any static equality-degenerate branches with the same thermodynamic volume
and the same entropy as $\Xref(V)$.  Such branches are allowed saturation
cases and are not violations of the theorem.

We assume that the AdS charges are finite and integrable, and that the horizons are smooth and bifurcate. The entropy is taken to be the stationary Iyer--Wald entropy, while \(P\) and \(\Vth\) are the pressure and thermodynamic volume appearing in the extended first law. The space $\scrC_V$ contains only exact stationary black holes.  For a
global comparison, one can choose affine branch coordinates $q^A$ on a stationary
fixed-volume branch $\scrB_V\subset\scrC_V$ connecting $\Xref(V)$ to the
solution $X$.  We parameterize the comparison by the affine segment
\begin{equation}
q^A(s)=q_0^A+s\,\Delta q^A,\qquad 0\le s\le1,
\label{eq:affine_comparison}
\end{equation}
which lies in $\scrB_V$, starts at $\Xref(V)$, and ends at $X$.  Along this exact
stationary family the bulk canonical energy vanishes.  We require the
stationary boundary Hessian
\begin{equation}
\calB_{\infty,V}^{\rm stat}(\Delta q,\Delta q)
=-T_s\,\frac{\partial^2 \SW}{\partial q^A\partial q^B}
\,\Delta q^A\Delta q^B
\ge0
\label{eq:component_positivity_input}
\end{equation}
throughout the segment.  Here the derivatives are taken with respect to the chosen affine coordinates along the stationary branch. Thus $\mathcal{B}_{\infty,V}^{\rm stat}$ is minus the temperature times the fixed-volume entropy Hessian evaluated along that branch.  This non-negativity is the input to the global theorem.
We do not claim that it holds universally. Rather, we have verified its sign explicitly for the Kerr--AdS and RN--AdS stationary examples discussed below.  To characterize the equality case, we further require that the comparison segment for which
$\calB_{\infty,V}^{\rm stat}(\Delta q,\Delta q)=0$ at every interior point
lies entirely within the allowed equality sector $\scrZ_V$.

\emph{\textbf{Theorem}.---}
For every solution $X$ on the chosen stationary fixed-volume branch, using
the affine comparison above,
\begin{equation}
\SW(X)\leq \SW\!\left(\Xref(V)\right).
\label{eq:theorem_main}
\end{equation}
If, in addition, zero-boundary-Hessian comparison segments are confined to
$\scrZ_V$, equality occurs only on the allowed saturation sector.  We only need the
chosen fixed-volume stationary comparison and other stronger assumptions such as connectedness, compactness, or properness of the full stationary solution space are not required.

\emph{\textbf{Fixed-volume entropy Hessian}.---}
On $\scrC_V$ define the constrained entropy functional
\begin{equation}
\SV[X]=\SW[X]-\mu\Vth[X],
\label{eq:constrainedS}
\end{equation}
where $\mu$ is the Lagrange multiplier implementing the fixed-volume constraint.  At the reference solution,
\begin{equation}
\delta\SV\big|_{\Xref}=0.
\label{eq:criticalX0}
\end{equation}

Before proceeding further, it is important to distinguish the tangent space of the stationary solution
manifold from the larger space of linearized perturbations.  We denote by
$\scrT_{X,V}$ the constrained linearized perturbation space over a stationary
background $X$, with $\delta\Vth=0$ and the chosen AdS boundary conditions,
and by $T_X\scrC_V\subset\scrT_{X,V}$ its finite-dimensional stationary
subspace.  For $v\in\scrT_{X,V}$ the bulk Hollands--Wald canonical energy is
\cite{Hollands:2012sf}
\begin{equation}
\Ebulk_{\chi,X}(v,v)=
\int_\Sigma \bomega(\ph;v,\Lie_\chi v).
\label{eq:bulkE}
\end{equation}
The boundary-completed quadratic form on $\scrT_{X,V}$ is
\begin{equation}
\Etot_{\chi,V}(v,v)
=
\Ebulk_{\chi,X}(v,v)+\mathcal B_{\infty,V}(v,v).
\label{eq:Etot_def}
\end{equation}
This is a stability quadratic form on the full constrained perturbation space.
The entropy Hessian, by contrast, is defined on the stationary manifold.
For a stationary tangent $v_{\rm stat}\in T_X\scrC_V$,
$\Lie_\chi v_{\rm stat}=0$, and therefore
\begin{equation}
\begin{split}
&\Ebulk_{\chi,X}(v_{\rm stat},v_{\rm stat})=0,\\&
\mathcal B_{\infty,V}(v_{\rm stat},v_{\rm stat})
=-T_X\,\mathrm{Hess}_V(\SW)(v_{\rm stat},v_{\rm stat}).
\end{split}
\label{eq:hessian}
\end{equation}
At the critical reference saddle, this gives an intrinsic local Hessian
statement (see Appendix~\ref{app:covariant_hessian} for the covariant second-variation origin).  Hence a local fixed-volume entropy maximum requires only
non-negativity of the stationary boundary quadratic form,
\begin{equation}
\mathcal B_{\infty,V}(v_{\rm stat},v_{\rm stat})\ge0,
\qquad v_{\rm stat}\in T_{\Xref}\scrC_V.
\label{eq:positiveE_local}
\end{equation}
For charge-free non-stationary perturbations, positivity of the Hollands--Wald bulk canonical energy is the relevant dynamical-stability condition. We have summarized the corresponding perturbation split and standard bulk-positivity discussion in Appendix~\ref{app:dynamical_sector}. The fixed-volume RII is related to exact stationary black-hole families, so its local entropy Hessian and global affine comparison involve only $v_{\rm stat}\in T_X\scrC_V$.  Since $\Lie_\chi v_{\rm stat}=0$, the bulk canonical energy vanishes along these directions and the comparison is governed by the stationary boundary charge Hessian.

\emph{\textbf{Stationary boundary contribution}.---}
Along an exact stationary black-hole family,
\begin{equation}
\Lie_\chi v_{\rm stat}=0,
\qquad
\Ebulk_{\chi,X}(v_{\rm stat},v_{\rm stat})=0,
\end{equation}
so that
\begin{equation}
\begin{split}
\Etot_{\chi,V}(v_{\rm stat},v_{\rm stat})
={}&\mathcal B_{\infty,V}(v_{\rm stat},v_{\rm stat})\\
={}&-T\,\mathrm{Hess}_V(\SW)(v_{\rm stat},v_{\rm stat}).
\label{eq:stat_block_thermo_positive_main}
\end{split}
\end{equation}
Equation~\eqref{eq:stat_block_thermo_positive_main} determines the stationary contribution once the fixed-volume entropy Hessian along the relevant stationary branch is known. It does not imply a universal sign.
As stated earlier, we do not claim a general proof of $\mathcal B_{\infty,V}\ge0$.  Instead, its sign is
checked explicitly in the stationary examples considered here: the Kerr--AdS fixed-volume expansion gives a non-negative boundary contribution (strictly
positive away from the Schwarzschild endpoint in the small-rotation regime),
while the static RN--AdS direction gives the saturation value $\mathcal B_{\infty,V}=0$ because fixed $V_{\rm th}$ fixes $r_+$ and hence
$S$ (see Appendix~\ref{app:stationary_checks}).

\begin{proposition}[Affine stationary comparison]
\label{prop:pathwise_positive}
Let $\scrB_V\subset\scrC_V$ be a stationary fixed-volume branch with affine
coordinates $q^A$, and let
$q^A(s)=q_0^A+s\,\Delta q^A$ be a comparison segment from $\Xref(V)$ to
$X\in\scrB_V$.  Suppose
\begin{equation}
\calB_{\infty,V}^{\rm stat}(\Delta q,\Delta q)\ge0
\qquad (0<s<1)
\label{eq:pathwise_positive_prop}
\end{equation}
along the segment.  Then $\SW[\gamma(s)]$ is concave in the affine parameter
$s$, and
\begin{equation}
\SW(X)\le \SW(\Xref(V)).
\end{equation}
\end{proposition}

\begin{proof}
Because the comparison segment consists of exact stationary black holes,
$\Ebulk_{\chi_s}=0$.  Affinity of the branch coordinates gives
\begin{equation}
\frac{\dd^2}{\dd s^2}\SW[\gamma(s)]
=
\frac{\partial^2\SW}{\partial q^A\partial q^B}
\Delta q^A\Delta q^B
=
-\frac{1}{T_s}\,
\calB_{\infty,V}^{\rm stat}(\Delta q,\Delta q)\le0 .
\end{equation}
At the Schwarzschild--AdS reference point the first variation vanishes, so
$\dd\SW/\dd s|_{s=0}=0$.  Therefore $\SW[\gamma(s)]$ is non-increasing for
$s\ge0$, which gives the stated inequality.
\end{proof}

\emph{\textbf{Lemma} (local stationary boundary positivity at Schwarzschild--AdS).---}
Let \(\Xref(V)\) be a non-extremal Schwarzschild--AdS black hole with compact
spherical horizon and thermodynamic volume \(V\).  Consider stationary
fixed-volume tangents \(v_{\rm stat}\in T_{\Xref}\scrC_V\), with pure gauge
directions quotiented out.  For the stationary Kerr--AdS and static RN--AdS directions considered here,
\begin{equation}
\mathcal B_{\infty,V}(v_{\rm stat},v_{\rm stat})\ge0 .
\label{eq:local_Etot_positive_X0}
\end{equation}
Consequently,
\begin{equation}
\mathrm{Hess}_{\Xref}(\SV)(v_{\rm stat},v_{\rm stat})\le0,
\label{eq:local_entropy_hessian_X0}
\end{equation}
so \(\Xref(V)\) is a semidefinite local entropy maximum at fixed
thermodynamic volume along these stationary branches. The Kerr direction is marginal at quadratic order at the Schwarzschild endpoint, but the entropy decreases at fourth order, while the static RN--AdS direction is an allowed equality direction. The exact stationary checks are given in Appendix~\ref{app:stationary_checks}.

\emph{\textbf{Proof}.---}
Because \(v_{\rm stat}\) is tangent to an exact stationary
black-hole family,
\begin{equation}
\Lie_{\chi_0}v_{\rm stat}=0,
\qquad
\Ebulk_{\chi_0,\Xref}(v_{\rm stat},v_{\rm stat})=0 .
\label{eq:local_stat_bulk_zero}
\end{equation}
Hence the quadratic form relevant to the fixed-volume entropy Hessian reduces
to the stationary asymptotic boundary Hessian,
\begin{equation}
\begin{split}
\Etot_{\chi_0,V}(v_{\rm stat},v_{\rm stat})
={}&
\mathcal B_{\infty,V}(v_{\rm stat},v_{\rm stat})\\
={}&
-T_0\,\mathrm{Hess}_{\Xref}(\SV)(v_{\rm stat},v_{\rm stat}) .
\end{split}
\label{eq:local_stat_boundary}
\end{equation}
The fixed-volume constraint fixes the Schwarzschild radius, so variations in \(r_+\) are not allowed within the fixed-volume stationary sector. Along the Kerr--AdS branch at fixed thermodynamic volume, one finds
\begin{equation}
S(a)\big|_{V}
=
S(0)\big|_{V}
-
\frac{\pi}{12G}
\frac{(R^2+\ell^2)^2}{R^2\ell^4}a^4
+O(a^6),
\label{eq:local_kerr_fourth}
\end{equation}
and therefore
\begin{equation}
\mathcal B_{\infty,V}(\partial_a,\partial_a)
=
\frac{\pi T_0}{G}
\frac{(R^2+\ell^2)^2}{R^2\ell^4}a^2
+O(a^4)>0
\label{eq:local_kerr_positive_interior}
\end{equation}
for sufficiently small nonzero \(a\).  At the Schwarzschild endpoint the
quadratic term vanishes, so the Kerr direction is marginal there, while the
fourth-order entropy decrease shows that it does not generate an equality
branch.

For the static RN--AdS direction,
\begin{equation}
S_{\rm RN-AdS}=\frac{\pi r_+^2}{G},
\qquad
\Vth=\frac{4\pi r_+^3}{3} .
\label{eq:local_RN_equal}
\end{equation}
At fixed \(\Vth\), \(r_+\) is fixed, and hence
\begin{equation}
S_{\rm RN-AdS}(\Vth,Q)=S_{\rm Schw-AdS}(\Vth).
\label{eq:local_RN_static_equality}
\end{equation}
Thus the stationary boundary Hessian is semidefinite on the stationary
branches considered here, which gives Eqs.~\eqref{eq:local_Etot_positive_X0}
and~\eqref{eq:local_entropy_hessian_X0}.\qed

The boundary completion is essential for stationary comparisons.  If \(v\)
is tangent to a family of exact stationary black holes, then
\begin{equation}
\Lie_\chi v=0,
\qquad
\Ebulk_{\chi,X}(v,v)=0 .
\label{eq:Ebulk_stationary_zero_main}
\end{equation}
The entropy curvature along the stationary solution space is therefore
carried entirely by the asymptotic boundary term. This is precisely where the boundary completion becomes indispensable for the global theorem: without $\calB_{\infty,V}^{\rm stat}$, the canonical-energy functional would vanish identically along the stationary comparison and would contain no information capable of producing the reverse-isoperimetric entropy ordering. For an affine
fixed-\(\Vth\), fixed-\(P\), fixed-\(\lambda_I\) stationary family,
differentiating the extended first law twice gives
\begin{equation}
\begin{split}
T\,\frac{\dd^2\SW}{\dd s^2}={}&
\frac{\dd^2 M}{\dd s^2}
-\Omega_i\frac{\dd^2J_i}{\dd s^2}
-\Phi_\alpha\frac{\dd^2Q_\alpha}{\dd s^2}\\
&-\frac{\dd T}{\dd s}\frac{\dd\SW}{\dd s}
-\frac{\dd\Omega_i}{\dd s}\frac{\dd J_i}{\dd s}
-\frac{\dd\Phi_\alpha}{\dd s}\frac{\dd Q_\alpha}{\dd s}.
\end{split}
\label{eq:thermo_charge_hessian_main}
\end{equation}
The term proportional to \(\dd T/\dd s\) vanishes at an entropy-critical point but must be retained away from it. Along the affine stationary branch, the boundary Hessian is $\calB_{\infty,V}^{\rm stat}=-T\,\frac{\dd^2\SW}{\dd s^2}$. This is the quantity used in the global stationary comparison.

The relation to the usual canonical-energy literature is complementary.  The
bulk term \(\Ebulk_\chi\) is the dynamical part of
\(\Etot_{\chi,V}\).  In Einstein gravity, Hollands and Wald showed that
positivity of this bulk canonical energy on the appropriate constrained
subspace is the covariant criterion for linearized stability, and that it is
related to second variations of asymptotic charges and horizon area
\cite{Hollands:2012sf}.  Prabhu and Wald showed that negative canonical energy
leads to exponential growth in the corresponding dynamical problem
\cite{Prabhu:2015rua}.  For static Einstein--AdS black holes the dynamical,
non-stationary sectors are also supported by the standard positive-energy
master-field structure: Regge--Wheeler and Zerilli in four dimensions, and
Kodama--Ishibashi in higher dimensions
\cite{Regge:1957td,Zerilli:1970se,Kodama:2003jz,Ishibashi:2003ap}.  These
results motivate the bulk stability part, while Eq.~\eqref{eq:thermo_charge_hessian_main}
shows why the boundary term is the relevant stationary thermodynamic part.  

In higher-derivative or non-minimally coupled theories, positivity of
$\Etot_{\chi,V}$ is not automatic.  Extra massive, ghostlike, scalarized,
or non-minimally coupled matter modes may make either the bulk contribution
or the boundary charge Hessian indefinite. For the global
stationary inequality, however, only the stationary boundary block enters,
because the bulk term vanishes on exact stationary families. A violation of the RII therefore requires
$\mathcal{B}_{\infty,V}^{\rm stat}$ to become negative somewhere
along the corresponding fixed-volume stationary comparison to
$\Xref(V)$. This possibility is realized in known examples:
charged AdS planar black holes in Einstein--Horndeski--Maxwell gravity with
Horndeski axions violate the reverse isoperimetric inequality, including the
entropy-based version appropriate to Wald entropy \cite{Feng:2017jub}.

\emph{\textbf{Global proof and equality}.---}
The local lemma establishes the stationary boundary-Hessian statement at
the Schwarzschild--AdS saddle. The global comparison is likewise a
stationary statement.  Let $\scrB_V\subset\scrC_V$ admit affine branch
coordinates $q^A$, and let
$q^A(s)=q_0^A+s\,\Delta q^A$ be a fixed-volume stationary comparison segment
from $\Xref(V)$ to $X$.  Define
\begin{equation}
F(s)=\SW[\gamma(s)] .
\end{equation}
Because the path consists of exact stationary black holes,
$\Ebulk_{\chi_s}=0$. Thus the global entropy comparison is carried entirely by the boundary part of the boundary-completed canonical energy and the bulk Hollands--Wald term plays no role along this stationary segment.  By the definition of the stationary boundary Hessian
in the affine branch coordinates,
\begin{equation}
F''(s)=
-\frac{1}{T_s}\,
\calB_{\infty,V}^{\rm stat}(\Delta q,\Delta q)\le0 .
\label{eq:pathconcavity}
\end{equation}
With the affine parameterization, the second derivative along the stationary
comparison is directly the constrained stationary Hessian.  Since $F'(0)=0$, concavity gives $F(1)\le F(0)$, which is
Eq.~\eqref{eq:theorem_main}.

If equality holds, the non-negative boundary Hessian in
Eq.~\eqref{eq:pathconcavity} must vanish throughout the interior of the
comparison segment:
\begin{equation}
\calB_{\infty,V}^{\rm stat}(\Delta q,\Delta q)=0,
\qquad 0<s<1 .
\end{equation}
By the zero-boundary-Hessian condition in the theorem, the segment then lies in the allowed equality sector $\scrZ_V$. A direction that is marginal only at an endpoint, such as infinitesimal Kerr--AdS rotation at Schwarzschild--AdS, does not imply equality. Equality requires that the zero direction extend along a stationary segment on which the entropy remains unchanged.

This proves the global entropy comparison along the chosen affine
fixed-volume stationary segment whenever the boundary-Hessian sign holds.  Appendix~\ref{app:stationary_checks} gives the exact Kerr--AdS and RN--AdS stationary checks in the compact Einstein--AdS sector.

\emph{\textbf{Einstein limit and higher-derivative form}.---}
For Einstein gravity,
\begin{equation}
\SW=\frac{A}{4G},
\end{equation}
and the static AdS black hole obeys
\begin{equation}
\Vth=\frac{\omega_{D-2}r_h^{D-1}}{D-1},
\qquad
A_0=\omega_{D-2}r_h^{D-2} .
\end{equation}
Equation~\eqref{eq:theorem_main} gives
\begin{equation}
A\leq \omega_{D-2}
\left(\frac{(D-1)\Vth}{\omega_{D-2}}
\right)^{\frac{D-2}{D-1}},
\label{eq:areabound}
\end{equation}
which is exactly Eq.~\eqref{eq:rii}.  In higher-derivative theories the theorem gives instead
\begin{equation}
\calR_{\rm W}\equiv
\left[\frac{\SW(\Xref(\Vth))}{\SW(X)}
\right]^{1/(D-2)}\geq1,
\label{eq:waldRatio}
\end{equation}
the Wald-entropy version of reverse isoperimetry.  A higher-derivative solution can fail an area inequality while still obeying the physical entropy inequality; only in Einstein gravity do the two coincide.

\emph{\textbf{Conclusion and Discussion}.---}
We showed that reverse isoperimetry can be understood through a boundary-completed canonical-energy framework on the fixed-volume covariant phase space. The natural variables are the stationary Iyer--Wald entropy and the thermodynamic volume appearing as the pressure conjugate in the extended Iyer--Wald first law. The key quadratic form is
\begin{equation}
\Etot_{\chi,V}=\Ebulk_\chi+\mathcal B_{\infty,V},
\end{equation}
not the bulk canonical energy alone. Its bulk part controls dynamical perturbations, while the asymptotic charge Hessian controls stationary thermodynamic directions that are not captured by the symplectic integral alone. On the stationary solution manifold, the bulk term vanishes, so fixed-volume entropy concavity is controlled entirely by the stationary boundary Hessian. For the global theorem, we choose an affine parameterization of a fixed-volume stationary comparison from the Schwarzschild--AdS reference saddle to the solution of interest. The required input is the sign of the stationary boundary Hessian along this comparison. Since the bulk canonical energy vanishes on the stationary branch, the boundary Hessian is the only nonzero contribution and directly establishes the fixed-volume entropy inequality. The explicit compact Einstein--AdS analysis verifies that Schwarzschild--AdS is a local entropy maximum and checks the sign of the stationary boundary contribution on the relevant branches. For Kerr--AdS, the entropy decrease at fixed volume is entirely encoded in this boundary term. Known violations fit naturally into this branch-based interpretation. The clearest examples are the super-entropic black holes obtained from ultraspinning limits \cite{Hennigar:2014cfa,Hennigar:2015cja}. Their event horizons are non-compact, although the horizon area is finite after compactification of the ultraspinning azimuthal direction. They therefore do not lie on the ordinary compact fixed-volume stationary branch used in the Schwarzschild--AdS comparison. Their violation of the RII should accordingly be viewed as a statement about a different sector of the stationary solution space, rather than as a failure of the Noether-charge variational mechanism. Within a branch for which an affine fixed-volume comparison to $\Xref(V)$ exists, a violation instead requires the stationary boundary Hessian to become negative somewhere along the comparison. This second possibility is realized in Einstein--Horndeski--Maxwell gravity: charged AdS planar black holes supported by Horndeski axions violate the reverse isoperimetric inequality, including the entropy-based form appropriate to Wald entropy \cite{Feng:2017jub}. In the present framework, such examples correspond to a failure of non-negativity of the stationary boundary Hessian on the corresponding fixed-volume branch. By contrast, equality requires the stationary boundary Hessian to vanish along the comparison. The Schwarzschild--AdS saddle gives the reference entropy at fixed thermodynamic volume, while special RN--AdS or equal-charge static limits may also saturate $\mathcal R=1$ because their thermodynamic volume reduces to the geometric volume of a Euclidean ball in $\mathbb{R}^{d-1}$. These are allowed equality cases rather than violations of the RII.

Thus the RII is recast as a covariant stability criterion based on the asymptotic Noether-charge Hessian. In Einstein gravity it reproduces the area--volume inequality, while in higher-derivative AdS gravity it gives the corresponding Wald-entropy inequality whenever the stated stationary branch conditions hold.

\appendix

\section{Stationary checks in four-dimensional Einstein--AdS}
\label{app:stationary_checks}

For Kerr--AdS, the standard extended-thermodynamic quantities are
\begin{equation}
S=\frac{\pi(r_+^2+a^2)}{G\Xi},\qquad
\Xi=1-\frac{a^2}{\ell^2},
\end{equation}
\begin{equation}
M=\frac{(r_+^2+a^2)(1+r_+^2/\ell^2)}{2G r_+\Xi^2},
\qquad J=aM,
\end{equation}
\begin{equation}
\Omega=\frac{a(1+r_+^2/\ell^2)}{r_+^2+a^2},\qquad
\Vth=\frac{2\pi(r_+^2+a^2)(2r_+^2\ell^2+a^2\ell^2-r_+^2a^2)}{3\ell^2\Xi^2 r_+}
\label{eq:app_kerr_thermo}
\end{equation}
\cite{Caldarelli:1999xj,Gibbons:2004ai,Cvetic:2010jb}.  Let $R$ be the Schwarzschild--AdS horizon radius with the same thermodynamic volume, $\Vth=4\pi R^3/3$.  Solving Eq.~\eqref{eq:app_kerr_thermo} perturbatively at fixed $\Vth$ gives
\begin{equation}
r_+(a)=R-\frac{R^2+\ell^2}{2R\ell^2}a^2
-\frac{(R^2+\ell^2)^2}{6R^3\ell^4}a^4+O(a^6).
\label{eq:app_fixedV_rplus}
\end{equation}
Substitution into the entropy yields
\begin{equation}
S(a)\big|_V=\frac{\pi R^2}{G}
-\frac{\pi}{12G}\frac{(R^2+\ell^2)^2}{R^2\ell^4}a^4+O(a^6).
\label{eq:app_fixedV_entropy}
\end{equation}
The absence of an $a^2$ term is also visible directly from the first law.  At $a=0$,
\begin{equation}
\left.\frac{\dd^2M}{\dd a^2}\right|_0
=\left.\frac{\dd\Omega}{\dd a}\right|_0
 \left.\frac{\dd J}{\dd a}\right|_0
=\frac{(R^2+\ell^2)^2}{2GR\ell^4},
\end{equation}
so $\dd^2S/\dd a^2|_0=0$.  Thus rotation is marginal at quadratic order at the Schwarzschild endpoint, while Eq.~\eqref{eq:app_fixedV_entropy} gives the first nonzero entropy decrease at fourth order.  Along the stationary branch the bulk canonical energy vanishes and
\begin{equation}
\mathcal B_{\infty,V}^{\rm stat}(\partial_a,\partial_a)
=-T\frac{\dd^2S}{\dd a^2}
=\frac{\pi T_0}{G}\frac{(R^2+\ell^2)^2}{R^2\ell^4}a^2+O(a^4)>0
\end{equation}
for sufficiently small nonzero $a$, where $T_0=(1+3R^2/\ell^2)/(4\pi R)$.

The finite stationary comparison can be checked exactly.  With
\begin{equation}
A=\frac{4\pi(r_+^2+a^2)}{\Xi},
\end{equation}
substitution of Eq.~\eqref{eq:app_kerr_thermo} gives
\begin{equation}
36\pi\Vth^2-A^3=
\frac{16\pi^3a^4\ell^4(r_+^2+a^2)^2(\ell^2+r_+^2)^2}
{r_+^2(\ell-a)^4(\ell+a)^4}\ge0,
\label{eq:app_kerr_exact}
\end{equation}
for $r_+>0$ and $|a|<\ell$, with equality only for $a=0$.  Hence every nonzero Kerr--AdS rotation has strictly smaller entropy than Schwarzschild--AdS at the same thermodynamic volume, and the endpoint-marginal Kerr direction does not extend to an equality branch.

For static RN--AdS,
\begin{equation}
S_{\rm RN-AdS}=\frac{\pi r_+^2}{G},\qquad
\Vth=\frac{4\pi r_+^3}{3}.
\end{equation}
Fixed $\Vth$ therefore fixes $r_+$ and gives
\begin{equation}
S_{\rm RN-AdS}(\Vth,Q)=S_{\rm Schw-AdS}(\Vth).
\label{eq:app_rn_equal}
\end{equation}
If charge variations are admitted, this is an exact static equality sector while in a fixed-charge $Q=0$ sector only Schwarzschild--AdS remains.

\section{Covariant origin of the stationary boundary Hessian}
\label{app:covariant_hessian}

For fields $\phi$, the Iyer--Wald symplectic current is
\begin{equation}
\bomega(\phi;\delta_1\phi,\delta_2\phi)
=\delta_1\btheta(\phi,\delta_2\phi)-\delta_2\btheta(\phi,\delta_1\phi),
\end{equation}
and the Hollands--Wald bulk canonical energy is
\begin{equation}
\Ebulk_{\chi,X}(v,v)=\int_\Sigma\bomega(\phi;v,\Lie_\chi v).
\label{eq:app_bulk_energy}
\end{equation}
The second variation of the Hamiltonian identity has the form
\begin{equation}
\delta^2H_\chi=\Ebulk_{\chi,X}(v,v)+B_\infty-B_{\cal B},
\label{eq:app_second_variation}
\end{equation}
where the horizon contribution contains the second variation of the Wald entropy and the asymptotic contribution contains the second variations of the conserved charges \cite{Wald:1993nt,Iyer:1994ys,Hollands:2012sf}.  At fixed $\Vth$, fixed pressure and fixed external couplings, the surviving asymptotic terms define the constrained charge Hessian $\mathcal B_{\infty,V}$, so the boundary-completed form is
\begin{equation}
\Etot_{\chi,V}=\Ebulk_\chi+\mathcal B_{\infty,V}.
\end{equation}
For a tangent to an exact stationary family, $\Lie_\chi v_{\rm stat}=0$, hence Eq.~\eqref{eq:app_bulk_energy} vanishes and Eq.~\eqref{eq:app_second_variation} reduces to
\begin{equation}
\mathcal B_{\infty,V}(v_{\rm stat},v_{\rm stat})
=-T\,\mathrm{Hess}_V(S_{\rm W})(v_{\rm stat},v_{\rm stat}).
\label{eq:app_stationary_hessian}
\end{equation}
This is the covariant origin of the local stationary Hessian used in the Letter.  Along an affine stationary fixed-volume family $q^A(s)=q_0^A+s\Delta q^A$, differentiating the extended first law once more gives
\begin{equation}
\begin{split}
T\frac{\dd^2S_{\rm W}}{\dd s^2}={}&
\frac{\dd^2M}{\dd s^2}-\Omega_i\frac{\dd^2J_i}{\dd s^2}-\Phi_\alpha\frac{\dd^2Q_\alpha}{\dd s^2}\\
&-\frac{\dd T}{\dd s}\frac{\dd S_{\rm W}}{\dd s}
-\frac{\dd\Omega_i}{\dd s}\frac{\dd J_i}{\dd s}
-\frac{\dd\Phi_\alpha}{\dd s}\frac{\dd Q_\alpha}{\dd s}.
\end{split}
\label{eq:app_affine_hessian}
\end{equation}
Thus the stationary boundary Hessian along the comparison is
\begin{equation}
\mathcal B_{\infty,V}^{\rm stat}(\Delta q,\Delta q)
=-T\frac{\dd^2S_{\rm W}}{\dd s^2},
\end{equation}
including away from an entropy-critical point, where the $\dd T/\dd s$ term in Eq.~\eqref{eq:app_affine_hessian} must be retained. This is the quantity whose sign enters the global affine comparison.

\section{Dynamical canonical-energy sector}
\label{app:dynamical_sector}

This appendix summarizes the dynamical canonical-energy sector in compact four-dimensional Einstein--AdS gravity. The reverse-isoperimetric argument naturally concerns tangent directions of exact stationary black-hole families at fixed volume. The discussion here addresses the separate problem of dynamical stability.

We work in four-dimensional Einstein gravity with negative cosmological constant at fixed pressure
\begin{equation}
    P=\frac{3}{8\pi G\ell^2},
\end{equation}
and with compact spherical conformal boundary, smooth bifurcate Killing horizons, and reflecting AdS boundary conditions.  The stationary examples considered below are Schwarzschild--AdS, the neutral Kerr--AdS rotational branch, and the static RN--AdS branch when charge variations are included.  The Kerr--AdS and RN--AdS thermodynamic variables used here are the standard extended-thermodynamic variables of asymptotically AdS black holes \cite{Caldarelli:1999xj,Gibbons:2004ai,Cvetic:2010jb}.

The restrictions should be kept explicit.  We do not claim a positivity theorem for arbitrary higher-derivative theories, arbitrary matter couplings, noncompact horizons, scalarized branches, disconnected topological sectors, or the full coupled Einstein--Maxwell perturbation problem around a generic charged rotating background. The stationary Kerr--AdS and RN--AdS checks are given in Appendix~\ref{app:stationary_checks}, while the covariant second-variation origin of the boundary Hessian is given in Appendix~\ref{app:covariant_hessian}.

Let \(X_0(V)\) denote the Schwarzschild--AdS solution with thermodynamic volume \(V\), and let
\begin{equation}
    \scrC_V=\{X: \Vth(X)=V\}
\end{equation}
denote the fixed-volume sector considered below. In Einstein gravity the Iyer--Wald entropy reduces to the Bekenstein--Hawking area entropy \cite{Wald:1993nt,Iyer:1994ys},
\begin{equation}
    \SW=\frac{A}{4G}.
\end{equation}
The boundary-completed canonical energy used in the Letter is
\begin{equation}
    \Etot_{\chi,V}(v,v)
    =
    \Ebulk_{\chi,X}(v,v)
    +
    \calB_{\infty,V}(v,v),
    \label{eq:SM_completed_energy_intro}
\end{equation}
where \(\Ebulk_{\chi,X}\) is the Hollands--Wald bulk canonical energy \cite{Hollands:2012sf}, \(\calB_{\infty,V}\) is the constrained asymptotic charge Hessian at fixed thermodynamic volume, and \(\chi\) is the horizon-generating Killing field.

The completed quadratic form acts on the full constrained linearized
perturbation space \(\scrT_{X,V}\).  The entropy Hessian is instead defined
on the stationary solution manifold.  For
\(v_{\rm stat}\in T_X\scrC_V\),
\begin{equation}
    \calB_{\infty,V}(v_{\rm stat},v_{\rm stat})
    =
    -T_X\,{\rm Hess}_{V}(\SW)(v_{\rm stat},v_{\rm stat}),
    \label{eq:SM_hessian_identity_intro}
\end{equation}
because \(\Lie_\chi v_{\rm stat}=0\) and hence the bulk canonical energy
vanishes.  This is the local fixed-volume entropy-Hessian identity used in
the stationary sector.

\subsection{Dynamical and stationary perturbation sectors}
\label{app:dyn_stat_split}

The completed quadratic form acts on the constrained
linearized perturbation space \(\scrT_{X,V}\) around a stationary black hole
\(X\).  The stationary tangent space \(T_X\scrC_V\) is only a subspace of
\(\scrT_{X,V}\). Such perturbations contain more information than the finite-dimensional thermodynamic state space.  We separate
\begin{equation}
    v=v_{\rm dyn}+v_{\rm stat}+v_{\rm gauge}.
    \label{eq:SM_dyn_stat_split}
\end{equation}
Here \(v_{\rm gauge}\) is pure diffeomorphism or gauge motion and is quotiented out.  The stationary part is tangent to the exact stationary black-hole family through \(X\):
\begin{equation}
    v_{\rm stat}
    =
    \left.\frac{\dd}{\dd\epsilon}X(q^A+\epsilon\,\delta q^A)\right|_{\epsilon=0},
    \qquad
    \delta\Vth=0,
    \label{eq:SM_vstat_def}
\end{equation}
where \(q^A\) denotes the integrable asymptotic charges and stationary parameters.  Equivalently, \(v_{\rm stat}\) is the stationary zero mode chosen to reproduce the first-order charge variations of \(v\).  The dynamical remainder is charge-free:
\begin{equation}
    \delta M[v_{\rm dyn}]
    =
    \delta J_i[v_{\rm dyn}]
    =
    \delta Q_\alpha[v_{\rm dyn}]
    =
    \delta\Vth[v_{\rm dyn}]
    =0,
    \label{eq:SM_vdyn_chargefree}
\end{equation}
after imposing the fixed external couplings and boundary conditions.

It is important to note that this split is not required in the RII proof, which restricts directly to $v_{\rm stat}\in T_X\scrC_V$.  It is included only to distinguish the stationary thermodynamic sector from the separate dynamical stability problem. The stationary part satisfies
\begin{equation}
    \Lie_\chi v_{\rm stat}=0,
    \qquad
    \Ebulk_{\chi,X}(v_{\rm stat},v_{\rm stat})=0,
\end{equation}
so it does not contribute to the bulk symplectic integral. Its sign is controlled by the constrained asymptotic charge Hessian,
\begin{equation}
\begin{split}
    \Etot_{\chi,V}(v_{\rm stat},v_{\rm stat})
    ={}&
    \calB_{\infty,V}(v_{\rm stat},v_{\rm stat})
    \\
    ={}&
    -T\,\mathrm{Hess}_{V}(\SW)(v_{\rm stat},v_{\rm stat}).
\end{split}
\end{equation}
The dynamical part is not parametrized by equilibrium variables such as \((S,J_i,Q_\alpha,\Omega_i,\Phi_\alpha)\).  It describes genuine linearized fields around the background, and its sign is controlled by the Hollands--Wald canonical energy \cite{Hollands:2012sf},
\begin{equation}
    \Ebulk_{\chi,X}(v_{\rm dyn},v_{\rm dyn})\ge0.
\end{equation}
\subsection{Bulk positivity for dynamical perturbations}
\label{app:bulk_positivity}

For completeness, and independently of the RII argument, we next consider perturbations with vanishing variations of the stationary charges.  These are the genuinely dynamical perturbations.  Around a static spherical Einstein--AdS black hole, gravitational perturbations decompose into gauge-invariant Regge--Wheeler--Zerilli master fields, equivalently into the Kodama--Ishibashi master variables \cite{Regge:1957td,Zerilli:1970se,Kodama:2003jz,Ishibashi:2003ap}.  The master variables
\begin{equation}
    \Psi_{\sigma\ell m}(t,r),
    \qquad
    \sigma\in\{\mathrm{axial},\mathrm{polar}\},
\end{equation}
obey equations of the form
\begin{equation}
    -\partial_t^2\Psi_\sigma
    +
    \partial_{r_*}^2\Psi_\sigma
    -
    V_\sigma(r)\Psi_\sigma=0,
    \qquad
    \frac{\dd r_*}{\dd r}=f^{-1}(r),
    \label{eq:SM_master_equation}
\end{equation}
with
\begin{equation}
    f(r)=1-\frac{2GM}{r}+\frac{r^2}{\ell^2}.
\end{equation}
The physical modes satisfy \(\ell\ge2\).  The \(\ell=0\) and \(\ell=1\) sectors correspond to stationary parameter shifts or pure-gauge modes.

For reflecting AdS boundary conditions and regularity at the future horizon, the standard master-field stability construction rewrites the bulk canonical energy as a positive sum after an \(S\)-deformation \cite{Kodama:2003jz,Ishibashi:2003ap,Hollands:2012sf}.  Define
\begin{equation}
    D_\sigma=\partial_{r_*}+S_\sigma(r),
\end{equation}
and
\begin{equation}
    \widetilde V_\sigma(r)
    =
    V_\sigma(r)
    +
    \frac{\dd S_\sigma}{\dd r_*}
    -
    S_\sigma^2(r).
    \label{eq:SM_s_deformed_potential}
\end{equation}
For the compact spherical Einstein--AdS sector, the Regge--Wheeler--Zerilli/Kodama--Ishibashi construction gives a regular deformation such that
\begin{equation}
    \widetilde V_\sigma(r)\ge0
\end{equation}
for physical gravitational modes.  The bulk canonical energy then has the schematic positive form
\begin{equation}
\begin{split}
    \Ebulk_{\chi,X}(v,v)
    ={}&
    \frac12
    \sum_{\sigma,\ell,m}
    \int_{r_+}^{\infty}\dd r_*
    \bigg[
        |\partial_t\Psi_{\sigma\ell m}|^2
        +
        |D_\sigma\Psi_{\sigma\ell m}|^2
    \\
    &\hspace{3.6cm}
        +
        \widetilde V_\sigma|\Psi_{\sigma\ell m}|^2
    \bigg],
    \label{eq:SM_bulk_positive}
\end{split}
\end{equation}
with no negative contribution from the AdS boundary under reflecting boundary conditions.  Hence
\begin{equation}
    \Ebulk_{\chi,X}(v,v)\ge0.
    \label{eq:SM_bulk_nonnegative}
\end{equation}
Equality in Eq.~\eqref{eq:SM_bulk_positive} requires
\begin{equation}
    \partial_t\Psi_{\sigma\ell m}=0,
    \qquad
    D_\sigma\Psi_{\sigma\ell m}=0,
    \qquad
    \widetilde V_\sigma^{1/2}\Psi_{\sigma\ell m}=0.
    \label{eq:SM_zero_bulk_conditions}
\end{equation}
Together with regularity at the horizon and normalizable reflecting AdS falloff, the only physical solutions of Eq.~\eqref{eq:SM_zero_bulk_conditions} are trivial master fields. Thus a zero-bulk-energy dynamical perturbation is pure gauge, up to stationary zero modes associated with changes of mass, angular momentum, or charge.  At fixed \(\Vth\) and \(P\), and within a fixed external charge ensemble, the mass is not an independent variation. The remaining angular-momentum and charge variations correspond to stationary boundary directions.

\vspace{5em}
\bibliographystyle{utphys1}
\bibliography{bib}

@article{Kastor:2009wy,
    author = "Kastor, David and Ray, Sourya and Traschen, Jennie",
    title = "{Enthalpy and the Mechanics of AdS Black Holes}",
    eprint = "0904.2765",
    archivePrefix = "arXiv",
    primaryClass = "hep-th",
    doi = "10.1088/0264-9381/26/19/195011",
    journal = "Class. Quant. Grav.",
    volume = "26",
    pages = "195011",
    year = "2009"
}

@article{Dolan:2010ha,
    author = "Dolan, Brian P.",
    title = "{The cosmological constant and the black hole equation of state}",
    eprint = "1008.5023",
    archivePrefix = "arXiv",
    primaryClass = "gr-qc",
    reportNumber = "DIAS-STP-10-10",
    doi = "10.1088/0264-9381/28/12/125020",
    journal = "Class. Quant. Grav.",
    volume = "28",
    pages = "125020",
    year = "2011"
}

@article{Cvetic:2010jb,
    author = "Cvetic, M. and Gibbons, G. W. and Kubiznak, D. and Pope, C. N.",
    title = "{Black Hole Enthalpy and an Entropy Inequality for the Thermodynamic Volume}",
    eprint = "1012.2888",
    archivePrefix = "arXiv",
    primaryClass = "hep-th",
    reportNumber = "DAMTP-2010-125, MIFPA-10-56",
    doi = "10.1103/PhysRevD.84.024037",
    journal = "Phys. Rev. D",
    volume = "84",
    pages = "024037",
    year = "2011"
}

@article{Wald:1993nt,
    author = "Wald, Robert M.",
    title = "{Black hole entropy is the Noether charge}",
    eprint = "gr-qc/9307038",
    archivePrefix = "arXiv",
    reportNumber = "EFI-93-42",
    doi = "10.1103/PhysRevD.48.R3427",
    journal = "Phys. Rev. D",
    volume = "48",
    number = "8",
    pages = "R3427--R3431",
    year = "1993"
}

@article{Iyer:1994ys,
    author = "Iyer, Vivek and Wald, Robert M.",
    title = "{Some properties of Noether charge and a proposal for dynamical black hole entropy}",
    eprint = "gr-qc/9403028",
    archivePrefix = "arXiv",
    doi = "10.1103/PhysRevD.50.846",
    journal = "Phys. Rev. D",
    volume = "50",
    pages = "846--864",
    year = "1994"
}

@article{Kumar:2025mvj,
    author = "Kumar, Naman",
    title = "{A proof of the reverse isoperimetric inequality using a geometric-analytic approach}",
    eprint = "2508.13235",
    archivePrefix = "arXiv",
    primaryClass = "gr-qc",
    month = "8",
    year = "2025"
}

@article{Xiao:2023lap,
    author = "Xiao, Yong and Tian, Yu and Liu, Yu-Xiao",
    title = "{Extended Black Hole Thermodynamics from Extended Iyer-Wald Formalism}",
    eprint = "2308.12630",
    archivePrefix = "arXiv",
    primaryClass = "gr-qc",
    doi = "10.1103/PhysRevLett.132.021401",
    journal = "Phys. Rev. Lett.",
    volume = "132",
    number = "2",
    pages = "021401",
    year = "2024"
}

@article{Hollands:2012sf,
    author = "Hollands, Stefan and Wald, Robert M.",
    title = "{Stability of Black Holes and Black Branes}",
    eprint = "1201.0463",
    archivePrefix = "arXiv",
    primaryClass = "gr-qc",
    doi = "10.1007/s00220-012-1638-1",
    journal = "Commun. Math. Phys.",
    volume = "321",
    pages = "629--680",
    year = "2013"
}

@article{Prabhu:2015rua,
    author = "Prabhu, Kartik and Wald, Robert M.",
    title = "{Black Hole Instabilities and Exponential Growth}",
    eprint = "1501.02522",
    archivePrefix = "arXiv",
    primaryClass = "gr-qc",
    doi = "10.1007/s00220-015-2446-1",
    journal = "Commun. Math. Phys.",
    volume = "340",
    number = "1",
    pages = "253--290",
    year = "2015"
}

@article{Kodama:2003jz,
    author = "Kodama, Hideo and Ishibashi, Akihiro",
    title = "{A Master equation for gravitational perturbations of maximally symmetric black holes in higher dimensions}",
    eprint = "hep-th/0305147",
    archivePrefix = "arXiv",
    doi = "10.1143/PTP.110.701",
    journal = "Prog. Theor. Phys.",
    volume = "110",
    pages = "701--722",
    year = "2003"
}

@article{Ishibashi:2003ap,
    author = "Ishibashi, Akihiro and Kodama, Hideo",
    title = "{Stability of higher dimensional Schwarzschild black holes}",
    eprint = "hep-th/0305185",
    archivePrefix = "arXiv",
    doi = "10.1143/PTP.110.901",
    journal = "Prog. Theor. Phys.",
    volume = "110",
    pages = "901--919",
    year = "2003"
}

@article{Feng:2017jub,
    author = "Feng, Xing-Hui and Liu, Hai-Shan and Lu, Wen-Tian and Lu, H.",
    title = "{Horndeski Gravity and the Violation of Reverse Isoperimetric Inequality}",
    eprint = "1705.08970",
    archivePrefix = "arXiv",
    primaryClass = "hep-th",
    doi = "10.1140/epjc/s10052-017-5356-x",
    journal = "Eur. Phys. J. C",
    volume = "77",
    number = "11",
    pages = "790",
    year = "2017"
}

@article{Hennigar:2014cfa,
    author = "Hennigar, Robie A. and Kubiz{\v{n}}{\'a}k, David and Mann, Robert B.",
    title = "{Entropy Inequality Violations from Ultraspinning Black Holes}",
    eprint = "1411.4309",
    archivePrefix = "arXiv",
    primaryClass = "hep-th",
    doi = "10.1103/PhysRevLett.115.031101",
    journal = "Phys. Rev. Lett.",
    volume = "115",
    number = "3",
    pages = "031101",
    year = "2015"
}

@article{Hennigar:2015cja,
    author = "Hennigar, Robie A. and Kubiz{\v{n}}{\'a}k, David and Mann, Robert B. and Musoke, Nathan",
    title = "{Ultraspinning limits and super-entropic black holes}",
    eprint = "1504.07529",
    archivePrefix = "arXiv",
    primaryClass = "hep-th",
    doi = "10.1007/JHEP06(2015)096",
    journal = "JHEP",
    volume = "06",
    pages = "096",
    year = "2015"
}

@article{Caldarelli:1999xj,
    author = "Caldarelli, Marco M. and Cognola, Guido and Klemm, Dietmar",
    title = "{Thermodynamics of Kerr-Newman-AdS black holes and conformal field theories}",
    eprint = "hep-th/9908022",
    archivePrefix = "arXiv",
    reportNumber = "UTF-434",
    doi = "10.1088/0264-9381/17/2/310",
    journal = "Class. Quant. Grav.",
    volume = "17",
    pages = "399--420",
    year = "2000"
}

@article{Gibbons:2004ai,
    author = "Gibbons, G. W. and Perry, M. J. and Pope, C. N.",
    title = "{The First law of thermodynamics for Kerr-anti-de Sitter black holes}",
    eprint = "hep-th/0408217",
    archivePrefix = "arXiv",
    reportNumber = "DAMTP-2004-87, MIFP-04-17",
    doi = "10.1088/0264-9381/22/9/002",
    journal = "Class. Quant. Grav.",
    volume = "22",
    pages = "1503--1526",
    year = "2005"
}

@article{Regge:1957td,
    author = "Regge, Tullio and Wheeler, John A.",
    title = "{Stability of a Schwarzschild singularity}",
    doi = "10.1103/PhysRev.108.1063",
    journal = "Phys. Rev.",
    volume = "108",
    pages = "1063--1069",
    year = "1957"
}

@article{Zerilli:1970se,
    author = "Zerilli, Frank J.",
    title = "{Effective potential for even parity Regge-Wheeler gravitational perturbation equations}",
    doi = "10.1103/PhysRevLett.24.737",
    journal = "Phys. Rev. Lett.",
    volume = "24",
    pages = "737--738",
    year = "1970"
}

\end{document}